\documentclass[11pt]{article}
\usepackage[letterpaper,margin=1in]{geometry}

\usepackage[utf8]{inputenc}

\usepackage{amsmath,amssymb,amsthm,amsfonts}
\usepackage{mathtools}
\usepackage{graphicx} 
\usepackage{float}
\usepackage{subfigure} 
\usepackage{arydshln}
\usepackage{multirow}
\usepackage{authblk,cite}
\usepackage{bm} 
\usepackage{mathrsfs}   
\usepackage{setspace}
\usepackage{url}
\usepackage{hyperref}
\hypersetup{colorlinks,linkcolor=blue,urlcolor=blue,citecolor=blue}
\usepackage{url}
\usepackage{xcolor}
\usepackage[title]{appendix}
\usepackage{qcircuit}

\usepackage{tikz}  
\usepackage{tikz-cd}  
\usetikzlibrary{positioning}
\usepackage{tikz-qtree}
\usepackage{mathdots}
\usepackage{authblk}

\usepackage{qcircuit}
\newcommand{\ket}[1]{| #1 \rangle}
\newcommand{\bra}[1]{\langle #1|}

\newcommand{\ketbra}[2]{| #1 \rangle \langle #2 |}

\usepackage{subfigure}
\usepackage{subcaption}

\DeclareMathOperator{\poly}{poly}

\DeclareMathOperator{\tr}{Tr}

\DeclareMathOperator{\sgn}{sgn}

\renewcommand{\Re}{\operatorname{Re}}

\def\x{\boldsymbol{x}}

\allowdisplaybreaks

\newtheorem{thm}{Theorem}[section]
\newtheorem{prop}[thm]{Proposition}

\newtheorem{conj}[thm]{Conjecture}
\newtheorem{prob}[thm]{Problem}

\theoremstyle{definition}

\newtheorem{definition}[thm]{Definition}
\newtheorem{rmk}[thm]{Remark}

\numberwithin{equation}{section}

\newtheorem*{thm*}{Theorem}
\newtheorem*{lem*}{Lemma}
\newtheorem*{prop*}{Proposition}
\newtheorem*{defn*}{Definition}
\newtheorem*{prob*}{Problem}
\newtheorem*{ques*}{Question}

\newcommand{\be}{\begin{equation}}
\newcommand{\ee}{\end{equation}}
\newcommand{\bes}{\begin{equation*}}
\newcommand{\ees}{\end{equation*}}
\newcommand{\bea}{\begin{eqnarray}}
\newcommand{\eea}{\end{eqnarray}}
\newcommand{\beas}{\begin{eqnarray*}}
\newcommand{\eeas}{\end{eqnarray*}}
\newcommand{\bal}{\begin{aligned}}
\newcommand{\eal}{\end{aligned}}

\usepackage[title]{appendix}

\title{ DQC1-completeness of normalized trace estimation for functions of log-local Hamiltonians }

\author[1]{Zhengfeng Ji\thanks{\texttt{jizhengfeng@tsinghua.edu.cn}}} 
\author[2,3]{Tongyang Li\thanks{\texttt{tongyangli@pku.edu.cn}}}
\author[4]{Changpeng Shao\thanks{\texttt{changpeng.shao@amss.ac.cn}}}
\author[2,3]{Xinzhao Wang\thanks{\texttt{wangxz@stu.pku.edu.cn}}}
\author[4]{Yuxin Zhang\thanks{\texttt{zhangyuxin@amss.ac.cn}}}

\affil[1]{\small{Department of Computer Science and Technology, Tsinghua University, Beijing, China}}

\affil[2]{Center on Frontiers of Computing Studies, Peking University, Beijing,  China}
\affil[3]{School of Computer Science, Peking University, Beijing, China}
\affil[4]{\small{SKLMS, Academy of Mathematics and Systems Science, Chinese Academy of Sciences, Beijing, China}}

\date{\today}

\begin{document}

\maketitle

\begin{abstract}

We study the computational complexity of estimating the normalized trace $2^{-n}\tr[f(A)]$ for a log-local Hamiltonian $A$ acting on $n$ qubits. This problem arises naturally in the DQC1 model, yet its complexity is only understood for a limited class of functions $f(x)$.

We show that if $f(x)$ is a continuous function with approximate degree $\Omega(\mathrm{poly}(n))$, then estimating $2^{-n}\tr[f(A)]$ up to constant additive error is DQC1-complete, under a technical condition on the polynomial approximation error of $f(x)$. This condition holds for a broad class of functions, including exponentials, trigonometric functions, logarithms, and inverse-type functions. We further prove that when $A$ is sparse, the classical query complexity of this problem is exponential in the approximate degree. Together, these results identify the approximate degree as the key parameter governing the complexity of normalized trace estimation: it characterizes both the quantum complexity (via efficient DQC1 algorithms) and the classical hardness, yielding an exponential quantum–classical separation. Our proof develops a unified framework that cleanly combines circuit-to-Hamiltonian constructions, periodic Jacobi operators, and tools from polynomial approximation theory, including the Chebyshev equioscillation theorem.

\vspace{.5cm}

{\em Note added.
Recently, with the assistance of GPT-5.6 Sol, we have further relaxed this ``technical condition'' to the assumption that $f(x)$ is Lipschitz continuous, thereby matching the regularity condition required by the quantum algorithm proposed in~\cite{cade}; see~\cite{shao-gpt} for further details. Thus, this subsumes all previous partial results towards DQC1-hardness and, at the same time, essentially closes the gap between the known classical hardness results and the regularity assumptions required by the corresponding quantum algorithm.

}

\end{abstract}





\section{Introduction}

Introduced by Knill and Laflamme in 1998~\cite{knill1998power}, the {\em One-Clean-Qubit model}, also known as {\em Deterministic Quantum Computation with One Qubit (DQC1)}, is a restricted model of quantum computation in which the input consists of a single pure qubit tensored with a maximally mixed $n$-qubit state. 
The computation starts from
\be
\rho = \ket{0} \bra{0} \otimes \frac{I_n}{2^n},
\ee
followed by a polynomial-size unitary circuit, after which the first qubit is measured in the computational basis. 
Despite its limited purity, DQC1 can efficiently solve certain problems for which no efficient classical algorithms are known. 
Moreover, it is believed to be strictly weaker than $\mathrm{BQP}$, and cannot be efficiently simulated classically unless the polynomial hierarchy collapses to the third level~\cite{morimae2014hardness}. 
These features make DQC1 a natural framework for studying intermediate quantum computational power.

A central computational task in this model is estimating the normalized trace $\frac{1}{2^n}\tr[U]$ of a polynomial-size quantum circuit~\cite{knill1998power,shepherd2006computation}, which is DQC1-complete. 
A range of other DQC1-complete problems can be viewed as variants of this task, including approximating the Jones polynomial~\cite{shor2008estimating}, estimating partition functions $\frac{1}{2^n}\tr[e^{-\beta A}]$~\cite{brandao2008entanglement}, and computing spectral quantities such as $\frac{1}{2^n}\tr[A^{-1}]$, $\frac{1}{2^n}\tr[A^{d}]$~\cite{cade, edenhofer2025dequantization}, and $\frac{1}{2^n}\tr[\log(A)]$~\cite{shao-wang}, among others~\cite{wang2024new, cerezo2020variational, khatri2019quantum, moulik2024dqc1, chowdhury2021computing}. 
These examples suggest that normalized trace estimation lies at the core of DQC1 complexity.

This perspective naturally leads to a general formulation: estimating spectral sums of the form $\frac{1}{2^n}\tr[f(A)]$, where $f(x)$ is defined on the spectrum of $A$. 
This raises the following fundamental question:
\begin{quote}
{\em What property of a function $f(x)$ makes the estimation of $\frac{1}{2^n}\tr[f(A)]$ computationally hard in the DQC1 model?}
\end{quote}
Existing hardness results~\cite{brandao2008entanglement, cade, edenhofer2025dequantization, shao-wang, chowdhury2021computing} are largely tailored to specific choices of $f$, and a general characterization remains elusive.

Addressing this question is important from several perspectives. 
On the complexity-theoretic side, it identifies the key structural features of functions that govern DQC1-hardness, and moves beyond function-specific analyses toward a unified understanding based on a single parameter of $f(x)$. 
From a broader viewpoint, the DQC1 model captures the power of quantum computation with only a single clean qubit and is widely regarded as a minimal model exhibiting nontrivial quantum advantage; thus, understanding normalized trace estimation sheds light on the fundamental capabilities and limitations of this model. 
On the practical side, spectral sums of the form $\tr[f(A)]$ arise ubiquitously in numerical linear algebra and machine learning. 
For instance, the log-determinant $\log\det(A)=\tr[\log(A)]$ plays a central role in marginal likelihood estimation in kernel methods and in training Gaussian processes or Gaussian graphical models~\cite{dong2017scalable, rasmussen2010gaussian}. 
More broadly, such quantities underpin tasks including trace estimation, partition function approximation, and matrix function evaluation. 
Efficient classical~\cite{boutsidis2017randomized, han2017approximating, musco_et_al:LIPIcs.ITCS.2018.8, edenhofer2025dequantization} and quantum~\cite{luongo2020quantum, zhao2019quantum, giovannetti2025quantum, cade} algorithms have been developed for these problems. 
Consequently, understanding when normalized spectral sum estimation becomes computationally hard helps delineate the fundamental limits of efficient classical computation, as well as the boundary between classical and quantum advantage.

\subsection{Main results}

Regarding the above fundamental problem, we approach it from two viewpoints: (i) establishing a general result for DQC1-completeness; and (ii) proving a lower bound on the classical query complexity. From our results, we see that the approximate degree is a key parameter in determining the computational complexity of normalized trace estimation.

To start, we define the following problem.

\begin{prob}[Normalized trace estimation problem]
    Let $f(x):[a,b]\subseteq [-1,1]\rightarrow [-1,1]$ be a continuous function, and let $A$ be a log-local Hamiltonian acting on $n$ qubits. In the normalized trace estimation problem, we aim to decide whether 
    \be
    \frac{1}{2^n}\tr[f(A)] \geq \frac{2}{3}
    \quad \text{or} \quad
    \frac{1}{2^n}\tr[f(A)] \leq \frac{1}{3}.
    \ee
    Equivalently, we aim to approximate $\frac{1}{2^n}\tr[f(A)]$ up to an additive error of at most $1/6$.
\end{prob}

Let $f(x): [a,b] \subseteq [-1,1]\rightarrow \mathbb{R}$ be a continuous function. For any integer $d$, we define 
\be
E_d=\min_{\deg(P)\leq d} ~~ \max_{x\in[a,b]} \|f(x)-P(x)\|
\ee
to be the best polynomial approximation error of degree at most $d$ over $[a,b]$.

\begin{thm}[DQC1-completeness]
\label{thm:intro-main}
Assume that $f(x):[a,b] \subseteq [-1,1]\rightarrow [-1,1]$ is a continuous function. Suppose there exists a constant $\varepsilon<1$ such that $d:=\widetilde{\deg}_\varepsilon (f)=\Omega(\poly(n))$ and $E_{d}/E_{d-1}<1/2$.
Then the ``normalized trace estimation problem" is DQC1-complete.
\end{thm}

We now briefly discuss the assumptions and implications of Theorem~\ref{thm:intro-main}.
Formal definitions of DQC1-completeness and approximate degree are given in Definitions~\ref{def:DQC1-complete} and~\ref{defn:approximate degree}.

\begin{itemize}
    \item First, the restriction that the spectrum of $A$ lies in $[-1,1]$ is without loss of generality, as we usually handle bounded/unitary operators on quantum computers. More generally, the theorem applies whenever the domain of $f(x)$ contains the spectrum of $A$. In addition, after an affine rescaling, the domain can be mapped to $[-1,1]$.

    \item Second, the log-locality assumption on $A$ is natural and necessary in the DQC1 setting. By definition, a log-local Hamiltonian is a sum of terms, each acting nontrivially on at most $O(\log n)$ qubits. This matches the computational power of DQC1, since DQC1 is equivalent to DQC$k$ for $k = O(\log n)$ clean qubits \cite{shor2008estimating}. In particular, quantum algorithms for estimating $\frac{1}{2^n}\tr[f(A)]$ in the DQC1 model typically require $O(\log n)$ ancillas when $A$ is log-local \cite{cade}.


    \item Third, the approximate-degree condition  $\widetilde{\deg}_\varepsilon(f) = \Omega(\poly(n))$ is natural and, although not stated explicitly, appears implicitly in prior works \cite{brandao2008entanglement, edenhofer2025dequantization, cade, shao-wang}. Indeed, since $f(x)$ can always be approximated by a polynomial of degree $\widetilde{\deg}_\varepsilon(f)$, a small approximate degree (e.g., $O(\log n)$) would allow $\tr[f(A)]$ to be computed efficiently by direct expansion, given that $A$ is log-local with ${\rm poly}(n)$ terms, e.g., see the classical algorithm given in \cite{edenhofer2025dequantization}. 

    \item
    Finally, the condition $E_d / E_{d-1} < 1/2$ arises in our analysis via a result from classical polynomial approximation theory (see Theorem~\ref{thm:approxiamtion}). It reflects a geometric separation between successive best-approximation errors and is satisfied by a broad class of functions commonly encountered in practice, including exponentials $e^{-\beta x}$, trigonometric function $\sin(tx), \cos(tx)$, logarithms $\log(1+\beta x)$, and inverse functions $1/\kappa x$ whenever the parameters  $\beta, \kappa $ are bounded away from degenerate regimes; see the discussion below Remark \ref{remark:key}. While this assumption is sufficient for our reduction, we do not expect it to be intrinsic. Indeed, in \cite{shao-gpt}, with the assistance of GPT, we relaxed this condition to Lipschitz continuity, thereby establishing DQC1-completeness for all Lipschitz continuous functions within our framework. Since the proof is already presented clearly in \cite{shao-gpt}, we do not reproduce it here.

\end{itemize}

As another result, we study the hardness from the viewpoint of query complexity. We now assume that $A$ is an $s$-sparse Hermitian matrix, i.e., each row has at most $s$ nonzero entries. We also assume standard oracle access to $A$ via the following two oracles:
\begin{itemize}
    \item[(1)] Given indices $(i,j)$, oracle $\mathcal{O}_1$ returns the $(i,j)$-th entry of $A$;
    \item[(2)] Given an index $i$, for any $j \leq s$, oracle $\mathcal{O}_2$ returns the position of the $j$-th nonzero entry in the $i$-th row.
\end{itemize}

In the course of proving a classical query lower bound, we need a lower bound for an analogue of the $k$-Forrelation problem in the DQC1 query model. More precisely, we consider an extension of the $k$-Forrelation problem~\cite{aaronson2015forrelation} in the DQC1 query model, defined as follows, where $N = 2^n$:
\begin{equation}
\label{k-forr Unitary}
{\tt Trace}_k(\x) :=
\frac{1}{N} \tr\!\left[
O_{\x_1} H^{\otimes n}
O_{\x_2} H^{\otimes n}
\cdots
H^{\otimes n}
O_{\x_k} H^{\otimes n}
\right].
\end{equation}
Here $\x_1, \ldots, \x_k \in \{\pm 1\}^N$ and $O_{\x_1},  \ldots, O_{\x_k}$ are oracles. In this problem, we aim to estimate ${\tt Trace}_k(\x)$ up to constant additive accuracy. This quantity can be viewed as a natural DQC1 analogue of the $k$-Forrelation problem, where acceptance probabilities are replaced by normalized traces.

In the DQC1 query model, we can estimate ${\tt Trace}_k(\x)$ to constant accuracy using $k$ quantum queries. Regarding the classical hardness,  we have the following result.

\begin{prop}
\label{prop:trace DQC1 complete-intro}
Estimating ${\tt Trace}_{k}$ up to constant additive accuracy is DQC1-complete. 
\end{prop}

Concerning the lower bound on classical query complexity, we currently have only limited observations.
On one hand, let $U$ denote the unitary in~\eqref{k-forr Unitary}. Then ${\tt Trace}_k(\x) = \mathbb{E}_{i \in [N]} \big[ \langle i | U | i \rangle \big],$
i.e., it is the average of the diagonal entries of $U$. In contrast, the $k$-Forrelation problem corresponds to estimating a single matrix element such as $\langle 0^n | U | 0^n \rangle$, for which a lower bound of $\widetilde{\Omega}(N^{1-1/k})$ is known~\cite{bansal2021k}. 
While ${\tt Trace}_k$ involves an average over $N$ such terms, these quantities are highly correlated through the shared oracle access, and it is unclear how to exploit this averaging to substantially reduce query complexity. This suggests that estimating ${\tt Trace}_k$ remains hard.
On the other hand, combining~\cite[Theorem~3.4]{bansal2021k} and~\cite[Theorem~1.5]{girish2026fourier}, we obtain that for $f(x)=\texttt{Trace}_k$, the $\ell_1$-norm of the level-$\ell$ Fourier coefficients satisfies $L_{1,\ell}(f) \le N^{(\ell-2)/2}$
up to lower-order factors. By~\cite[Corollary~3.5]{bansal2021k}, any decision tree of depth $d$ satisfies $L_{1,\ell} \le d^{\ell/2}$. Comparing these bounds suggests that $d^{\ell/2} \gtrsim N^{(\ell-2)/2} \text{ for all } \ell \le k,$
which in turn yields $d \gtrsim N^{1-2/k}$. 
Finally, using a similar technique \cite{bansal2021k}, one can easily prove a lower bound of $\widetilde{\Omega}(N^{1-2/k})$, with the exception that the additive error is of order $1/N$, which is too small.

The above motivates the following conjecture.

\begin{conj}
\label{conj:trace}
Estimating ${\tt Trace}_k$ up to an additive error $\varepsilon_k$ (depending only on $k$) requires $\widetilde{\Omega}(N^{1-2/k})$ classical queries.
\end{conj}  

Regarding Conjecture~\ref{conj:trace}, Shepherd made the following partial progress in~\cite[Proposition~3.3.6]{shepherd2010quantum}.

\begin{prop}
\label{prop: lower bound of trace 4}
Estimating ${\tt Trace}_4$ up to a constant additive error requires $\widetilde{\Omega}(N^{1/4})$ classical queries.
\end{prop}

With the above proposition, we have the following result.

\begin{thm}[Classical lower bound on query complexity]
\label{thm: intro lower bound}
Let $s$ be an integer. Assume that $f(x):[a,b] \subseteq [-1,1]\rightarrow [-1,1]$ is a continuous function.
Suppose there exists a constant $\varepsilon<1$ such that $d:=\widetilde{\deg}_\varepsilon (f)=\Omega(\poly(n))$ and $E_{d}/E_{d-1}<1/2$.
Then there is an $s$-sparse, log-local Hamiltonian $A$ for which
\be
\widetilde{\Omega} \left(
(s/2)^{(\widetilde{\deg}_\varepsilon(f)-4)/16}
\right)
\ee
classical queries are required to estimate the normalized trace of $f(A)$ up to a constant accuracy.
\end{thm}

As discussed in the abstract, the condition $E_{d}/E_{d-1}<1/2$ can be relaxed to Lipschitz condition in Theorems \ref{thm:intro-main} and \ref{thm: intro lower bound}. 
Taken together, Theorems~\ref{thm:intro-main} and \ref{thm: intro lower bound} identify the approximate degree of $f(x)$ as the fundamental parameter governing the computational hardness of normalized trace estimation in the DQC1 model.
More broadly, approximate degree has long played a central role in characterizing computational complexity, particularly for Boolean functions \cite{beals1998quantum, bun2022approximate, montanaro2024quantum},
and our result extends this perspective to the setting of quantum spectral sum estimation.




\subsection{Technical overview}

Our proof departs fundamentally from previous approaches, which are based on refinements of Brand{\~a}o’s framework \cite{brandao2008entanglement} using Taylor expansions together with delicate error control on spectral sums (see, e.g., \cite{edenhofer2025dequantization, shao-wang}). While these works combine circuit-to-Hamiltonian constructions with function-specific and technically involved analyses, our approach provides a unified and conceptually clean framework that applies to a broad class of functions through a single structural mechanism.

The starting point of our construction is still the circuit-to-Hamiltonian paradigm, but with the introduction of tunable parameters. Given an $n$-qubit quantum circuit $U = U_m \cdots U_2 U_1$ such that $m=\mathrm{poly}(n)$ and each $U_t$ acts on $O(\log n)$ qubits, we consider the log-local Hamiltonian
\begin{equation}
A =
\sum_{t=1}^m a_t \ket{t-1}\bra{t-1} \otimes I_n +
\sum_{t=1}^m b_t \bigl(
\ket{t}\bra{t-1} \otimes U_t +
\ket{t-1}\bra{t} \otimes U_t^\dagger
\bigr),
\end{equation}
where $a_t\in \mathbb{R}, b_t \in \mathbb{R}^*$ are free parameters and $\ket{m} := \ket{0}$. This Hamiltonian is log-local and consists of $\mathrm{poly}(n)$ terms. The dimension of $A$ is $m2^n$.

A key observation is that for any eigenpair $(e^{i\theta}, \, \ket{\phi_\theta})$ of $U$, the restriction of $A$ to the invariant subspace spanned by
\be
\Bigl\{
\ket{0} \otimes \ket{\phi_\theta}, \,\,\,\,
\ket{1} \otimes U_1 \ket{\phi_\theta}, \,\,\,\,
\ket{2} \otimes U_2U_1 \ket{\phi_\theta}, \,\,\,\,
\ldots, \,\,\,\,
\ket{m-1} \otimes U_{m-1}\cdots U_2U_1 \ket{\phi_\theta}
\Bigr\}
\ee
is a periodic Jacobi matrix of the form
\begin{equation}
\label{intro:periodic Jacobi matrix}
A_\theta :=
\begin{bmatrix}
a_1 & b_1 &        &        & b_m e^{i\theta} \\
b_1 & a_2 & b_2    &        &                \\
    & b_2 & \ddots & \ddots &                \\
    &     & \ddots & \ddots & b_{m-1}         \\
b_m e^{-i\theta} & & & b_{m-1} & a_m
\end{bmatrix}.
\end{equation}
As a consequence, the trace decomposes as
\begin{equation}
\label{intro:trace decomposition}
\tr[f(A)] =
\sum_{\substack{\theta :\, e^{i\theta} \text{ is an eigenvalue of } U}}
\tr[f(A_\theta)].
\end{equation}

Since estimating $\frac{1}{2^n}\Re(\tr[U]) = \frac{1}{2^n}\sum_\theta \cos\theta$ is DQC1-complete, to establish the DQC1 completeness of estimating the normalized trace $\tr[f(A)]$, it suffices to relate $\tr[f(A_\theta)]$ to $\cos\theta$ for any fixed $\theta$. This connection is established through classical results on periodic Jacobi matrices and the Chebyshev equioscillation theorem. 

More precisely, the characteristic polynomial of $A_\theta$ admits the representation
\be
\det(xI - A_\theta) = 2 b_1 \cdots b_m \bigl( \Delta(x) - \cos\theta \bigr),
\ee
where $\Delta(x)$, known as the discriminant, is independent of $\theta$. A direct calculation shows that $\tr[\Delta(A_\theta)]=m\cos\theta$. Moreover, for any $\theta$, the equation $\Delta(x) - \cos\theta = 0$ has $m$ real zeros, which are the eigenvalues of $A_\theta$.  

A natural intuition is that if $\Delta(x)$ closely approximates $f(x)$, then $\tr[f(A_\theta)]$ will correlate with $m \cos\theta$. However, this intuition must be refined: $\Delta(x)$ is an oscillatory polynomial that intersects any horizontal line $\cos\theta \in [-1,1]$ at exactly $m$ real points, a property that a generic function $f(x)$ usually does not possess. To exploit this structure, we invoke the Chebyshev equioscillation theorem \cite{golomb1962lectures}, which characterizes optimal min--max polynomial approximations. Roughly speaking, if $g(x)$ is the best uniform polynomial approximation of $f(x)$ of degree $d$, then $f(x)-g(x)$ oscillates with $d+2$ alternating points, see Theorem \ref{thm:Chebyshev Equioscillation Theorem} below.

We show that, after appropriate normalization, there is a discriminant $\Delta(x)$ that can be identified with an oscillatory approximation of $(f(x)-g(x))/E$, where $E=\|f-g\|_\infty$. Through $\Delta(x)$, we can reconstruct a periodic Jacobi matrix of the form \eqref{intro:periodic Jacobi matrix}. 
Moreover, $\deg(\Delta(x)) > \deg(g(x))$, which implies that $\tr[g(A_\theta)]$ is independent of $\cos\theta$ and can  be computed efficiently in ${\rm poly}(n)$ time.
Together with \eqref{intro:trace decomposition}, we obtain the desired correspondence between  $\frac{1}{m2^n} \tr[f(A)]$ and $\frac{1}{2^n} \Re(\tr[U])$. 

In this process, we can choose the degree of $g(x)$ as the approximate degree $\widetilde{\deg}_\varepsilon(f)$ for a fixed precision $\varepsilon$. In addition, we establish a connection between Chebyshev equioscillation theorem and the approximate degree of $f(x)$, revealing approximate degree as the fundamental parameter governing DQC1-hardness of normalized trace estimation.

Regarding the query complexity lower bound, we extend the $k$-Forrelation problem~\cite{aaronson2015forrelation} to the DQC1 model. Specifically, we consider the problem of estimating the normalized trace of 
$U = O_{\x_1} H^{\otimes n} \cdots O_{\x_k} H^{\otimes n},$
where $H$ is the $2 \times 2$ Hadamard gate and $O_{\x_1}, \ldots, O_{\x_k}$ are oracle operators. Combining the lower bound for estimating $\frac{1}{2^n}\tr[U]$ with the connection established above between $\frac{1}{m2^n} \tr[f(A)]$ and $\frac{1}{2^n} \Re(\tr[U])$, we obtain a lower bound for estimating $\frac{1}{m2^n} \tr[f(A)]$.


\subsection{Comparison with previous work \cite{edenhofer2025dequantization}}

Theorem 4.8 in \cite{edenhofer2025dequantization} establishes DQC1-hardness for estimating 
$\frac{1}{2^n}\tr[f(A)]$ in the block-encoding model. Concretely, it assumes the existence of a unitary $U$ with $O(\log n)$ ancillas such that its top-left corner approximates $A$, and that $f$ and $f^{-1}$ on $[-1/2,1/2]$ admit polynomial approximations of degree $O({\rm poly}(n/\varepsilon))$, with $f$ being Lipschitz continuous with Lipschitz constant ${\rm poly}(n)$. Under these assumptions, estimating $\frac{1}{2^n}\tr[f(A)]$ is DQC1-hard.

While this result provides a valuable characterization, the block-encoding assumption is somewhat restrictive for natural Hamiltonians. 
For instance, when $A = \sum_{j=1}^L \alpha_j P_j$ is a linear combination of Paulis, the LCU technique gives $\alpha = \sum_{j=1}^L |\alpha_j|$, and the block encoding requires $O(\log L)$ ancilla qubits. When $A$ is sparse, $\alpha$ equals the maximal number of non-zero entries in any row or column, and the block encoding uses $O(n)$ ancillas. In both cases, achieving $\alpha=1$ is nontrivial, even if $\|A\| \le 1$ is known, e.g., see \cite[page 23, arXiv version]{montanaro2024quantum}. Moreover, the requirement of $f^{-1}$ appears primarily as a technical condition in the proof, rather than reflecting an intrinsic hardness property of the problem. Finally, in this model, estimating $\frac{1}{2^n}\tr[A]$ is already DQC1-complete (see \cite[Theorem 4.6]{edenhofer2025dequantization}), whereas this is not automatically the case for log-local Hamiltonians.

In contrast, apart from the restriction $E_d/E_{d-1}<1/2$, our theorem \ref{thm:intro-main} establishes DQC1-hardness for log-local Hamiltonians under a natural and simple condition based solely on the approximate degree of $f$, without requiring $f^{-1}$ or additional assumptions. This yields a unified and broadly applicable framework for understanding the hardness of normalized trace estimation.

Additionally, in \cite[Theorem 1.1]{edenhofer2025dequantization}, the authors provide a classical algorithm for estimating the normalized trace of $f(A)$ for any $s$-sparse matrix $A$. Specifically, if $f(x)$ is a degree-$d$ polynomial satisfying $|f(x)| \le 1$ for all $x \in [-1,1]$, then there exists a classical algorithm with query complexity $O(s^d/\varepsilon^2)$ that outputs an $\varepsilon$-approximation of the normalized trace of $f(A)$. This result extends naturally to continuous functions via polynomial approximation. Since the complexity depends on the polynomial degree, one aims to approximate $f(x)$ by a polynomial of minimal degree, i.e., its approximate degree. 
In contrast, Theorem~\ref{thm: intro lower bound} establishes an exponential lower bound in terms of the approximate degree. Although a gap remains between the upper and lower bounds, both results indicate that the approximate degree is a key parameter governing the computational complexity of this problem.

\section{Preliminaries}
\label{intro:Preliminaries}

In this section, we list and also prove some results that are useful throughout this work.

\subsection{One clean qubit model}

The one clean qubit model (DQC1) is a restricted model of quantum computation that captures nontrivial quantum computational power with minimal quantum resources. In this model, the input consists of a single pure qubit (the ``clean" qubit) together with a collection of qubits in the maximally mixed state. A polynomial-size quantum circuit is then applied, and only the clean qubit is measured at the end of the computation, see Figure \ref{circuit for DQC1}.

A simple calculation shows that the probability of measuring zero is $\frac{1}{2} + \frac{1}{2} \frac{{\rm Re}(\tr[U])}{2^n}$. So the real part of the normalized trace can be estimated up to accuracy $\varepsilon$ by repeating the
procedure $O(1/\varepsilon^2)$ times. The imaginary part can be estimated similarly. This is a class of central problems that can be efficiently solved by the DQC1 mode, and is believed to be classically hard. 

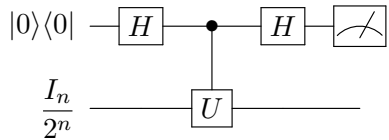
\begin{figure}[h]
\[
\Qcircuit @C=1em @R=1.5em {
\lstick{\ket{0}\bra{0}} &   \gate{H} & \ctrl{1}    & \gate{H}  & \meter \\
\lstick{\displaystyle \frac{I_n}{2^n}}   &  \qw       & \gate{U}   &  \qw  &  \qw  \\
}
\]
\caption{The DQC1 computational model.}
\label{circuit for DQC1}
\end{figure}

The study of DQC1 has led to the definition of the complexity class DQC1, which consists of problems efficiently solvable within this model. Understanding the power and limitations of DQC1 provides insight into the role of quantum coherence and entanglement in quantum computation, and helps clarify which quantum resources are essential for achieving quantum speedups.

\begin{definition}[DQC1 complexity class]
\label{def:DQC1-complete}
    A language $L$ is in DQC1 if and only if, for any $n$-bit string $x$, there exists a polynomial-time generated quantum circuit $U_x$ acting on $q(n)=\poly(n)$ qubits such that for some specified polynomial $p(n)$,
    \bea
    \label{eq:dqc1}
         \mu_{\mathrm{YES}} &\geq&  \frac{1}{2} + \frac{1}{p(n)}, \quad\text{if } x\in L, \\
        \mu_{\mathrm{YES}} &\leq&  \frac{1}{2} - \frac{1}{p(n)}, \quad\text{if } x\notin L.
    \eea
where
\be
\mu_{\mathrm{YES}} := \tr \Bigl[U_x \Bigl(\ketbra{0}{0} \otimes \frac{I_{q(n)-1}}{2^{q(n)-1}} \Bigr) U_x^{\dagger} \bigl(\ketbra{1}{1} \otimes I_{q(n)-1}\bigr) \Bigr]
\ee
is the probability of measuring the clean qubit at 1.
\end{definition}

In DQC1, we are allowed to repeat the computation a polynomial number of times, which is done outside the model, so the above is equivalent to the setting with $\mu_{\mathrm{YES}} \geq 2/3$ when $x\in L$ and  $\mu_{\mathrm{YES}} \leq 1/3$ when $x\not\in L$. 
This is because by repeating the experiment ${\rm poly}(r)$ times and taking the majority vote, one can achieve high probability of correctness \cite{shor2008estimating}.  


\begin{prop}[Theorem 2 of  \cite{knill1998power}]
\label{prop: TrU is DQC1 complete}
Let $\varepsilon\in (1/{\rm poly}(n) ,1)$, which can be a constant, and let $U$ be a unitary acting on $n$ qubits. Then
determining $\frac{1}{2^n} \Re(\tr[U]) > \varepsilon$ or $\frac{1}{2^n} \Re(\tr[U]) < - \varepsilon$ is DQC1-complete. Equivalently,
estimating $\frac{1}{2^n} \Re(\tr[U]) \pm \varepsilon$ is DQC1-complete.
\end{prop}

\begin{prop}[Corollary 6 of \cite{cade}]
\label{prop: in DQC1}
Assume that $A$ is a log-local Hamiltonian acting on $n$ qubits and $f(x):[-1,1]\rightarrow [-1,1]$ is Lipschitz continuous with Lipschitz constant $K$, then there is a DQC1 algorithm to estimate $\frac{1}{2^n} \tr[f(A)] \pm \varepsilon (K+1)$, where $\varepsilon=1/O({\rm poly}(n))$.
\end{prop}

The quantum algorithm in the above theorem is mainly based on quantum phase estimation, so the cost is linear in $1/\varepsilon$. When $K=O({\rm poly}(n))$, we can choose a smaller $\varepsilon$ such that $\varepsilon K$ is small. The key point here is that $1/\varepsilon, K$ have order $O({\rm poly}(n))$. For bounded Lipschitz continuous functions, we have $K=O(\widetilde{\deg}_\varepsilon(f)^2)$ via Markov brothers' inequality \cite{Markov}, where $\widetilde{\deg}_\varepsilon(f)$ is the approximate degree, see Definition \ref{defn:approximate degree}.


\subsection{Periodic Jacobi matrices}

We consider the following special type of Jacobi matrices, which is known as {\em periodic Jacobi matrices}, e.g., see \cite{boley1978matrix, van1976spectrum, boley1984modified, ferguson1980construction}:
\be \label{eq:periodic Jacobi matrix}
A_\theta :=\begin{bmatrix}
a_1    & b_1 &      & & b_m e^{i\theta} \\
b_1  & a_2   & b_2  & \\
     & b_2 &  \ddots   & \ddots \\
     &     & \ddots        & \ddots & b_{m-1}\\
b_m e^{-i\theta}        &&  & b_{m-1} & a_m \\
\end{bmatrix},
\ee
where $a_1,\ldots,a_m \in \mathbb{R}$, 
$b_1,\ldots,b_m \in \mathbb{R}^+$ and $\theta\in[0,2\pi)$. Obviously, $A_\theta$ is Hermitian, and so all the eigenvalues are real.

Let $B$ be the tridiagonal matrix of dimension $m$ with subdiagonal entries $b_1,\ldots,b_{m-1}$ (i.e., setting $b_m=0$ in $A_\theta$) and let $C$ be the tridiagonal matrix of dimension $m-2$ with subdiagonal entries $b_2,\ldots,b_{m-2}$  (i.e., removing the first and last row/column in $A_\theta$), then it is easy to compute that the characteristic polynomial of $A_\theta$ is (e.g., see \cite{boley1978matrix})
\be \label{char poly of A}
\det(xI-A_\theta)= h(x) - e \cos\theta ,
\ee
where
\bea
h(x) &=&  \det(xI-B) - b_m^2 \det(xI-C) , \label{1020} \\
e    &=& 2b_1b_2\cdots b_{m}.
\eea
The polynomial $\Delta(x) := h(x)/e$ 
is known as the {\em discriminant} of $A_\theta$ and plays an important role in the study of $A_\theta$. 
For example, assume $a_i=0, b_i=1/2$ for all $i$, then a simple calculation shows that $\Delta(x) = T_m(x)$, where $T_m(x)$ is the $m$-th Chebyshev polynomial of the first kind.

From \eqref{char poly of A} and the property that $A_\theta$ is Hermitian, we know that $\Delta(x) - \cos(\theta) = 0$ always has $m$ real zeros for any $\theta$. Based on this, it is not hard to see that $\Delta(x)$ is an oscillating function and it intersects any horizontal line $y=\cos(\theta)$ for $m$ times exactly. This is a property similar to the Chebyshev Equioscillation Theorem described below.

\subsection{ Chebyshev Equioscillation Theorem and approximate degree }

The Chebyshev Equioscillation Theorem is a central result in minimax approximation theory. It characterizes exactly when a polynomial is the unique minimax approximation to a continuous function.

\begin{thm}[Chebyshev Equioscillation Theorem \cite{golomb1962lectures}]
\label{thm:Chebyshev Equioscillation Theorem}
Let $f:[a,b] \rightarrow \mathbb{R}$ be a continuous function. Among all the polynomials of degree $\le d$, the polynomial $g$ minimizes the uniform norm of the difference $E:=\| f - g \| _\infty:=\max_{a\leq x\leq b}|f(x)-g(x)|$ if and only if there are $d+2$ points $a \le x_1 < x_2 < \cdots < x_{d+2} \le b$ such that 
\be
f(x_i) - g(x_i) = \sigma (-1)^i E
\ee
where $\sigma$ is either $-1$ or $+1$.
That is, the polynomial $g$ oscillates above and below $f$ at the interpolation points, and does so to the same degree.
\end{thm}

\begin{figure}[ht]
    \centering
    
    \scalebox{0.8}{
    \begin{minipage}[b]{0.45\textwidth}
        \centering
 \begin{tikzpicture}[>=stealth, thick]

\draw[->, line width=1pt] (0.3,-0.6) -- (7.5,-0.6);
\draw[->, line width=1pt] (0.7,-0.9) -- (0.7,5.6);

\draw[black, very thick, domain=1.3:6.4, samples=400, smooth]
  plot (\x, {0.6*\x + 1.6*sin((360*\x*\x)/12)});

\draw[black, very thick, domain=1.3:7, samples=400, smooth]
  plot (\x, {0.1 + ln(\x) + 0.27*\x});

\filldraw[black] (1.83703396902498,2.67152686357443736) circle (3pt);
\filldraw[black] (3.89561413282896,3.93059418951833095) circle (3pt);
\filldraw[black] (5.20878153406637,4.72148329977922110) circle (3pt);
\filldraw[black] (6.25375261213554,5.34962631925823562) circle (3pt);

\node (a1) at (1.83703396902498,2.67152686357443736+0.4) {$x_1$};
\node (a2) at (3.89561413282896,3.93059418951833095+0.4) {$x_3$};
\node (a3) at (5.20878153406637,4.72148329977922110+0.4) {$x_5$};
\node (a4) at (6.25375261213554,5.34962631925823562+0.4) {$x_7$};

\node (a4) at (6.25375261213554+0.9,5.34962631925823562-0.2) {{\color{black}$g(x)$}};
\node (a4) at (6.25375261213554+1.4,5.34962631925823562-1.3) {{\color{black}$f(x)$}};
\node (a4) at (1.83703396902498-0.3,2.67152686357443736-0.85) {{\color{black}$E$}};

\draw[black, thin] (1.83703396902498,2.67152686357443736) -- (1.83703396902498,1.204151469);
\draw[black, thin] (3.89561413282896,3.93059418951833095) -- (3.89561413282896,2.511667155);
\draw[black, thin] (5.20878153406637,4.72148329977922110) -- (5.20878153406637,3.156716972);
\draw[black, thin] (6.25375261213554,5.34962631925823562) -- (6.25375261213554,3.621694906);

\filldraw[black] (2.96116035016156,0.188438912811551340) circle (3pt);
\filldraw[black] (4.56618742632752,1.14464030196086264) circle (3pt);
\filldraw[black] (5.73416534303549,1.84362225581235317) circle (3pt);

\node (a2) at (2.96116035016156+0.1,0.188438912811551340-0.4) {$x_2$};
\node (a3) at (4.56618742632752+0.1,1.14464030196086264-0.4) {$x_3$};
\node (a4) at (5.73416534303549+0.1,1.84362225581235317-0.4) {$x_6$};

\draw[black, thin] (2.96116035016156,0.188438912811551340) -- (2.96116035016156,1.985094496);
\draw[black, thin] (4.56618742632752,1.14464030196086264) -- (4.56618742632752,2.851549200);
\draw[black, thin] (5.73416534303549,1.84362225581235317) -- (5.73416534303549,3.394666846);

\end{tikzpicture}
    \end{minipage}
    \quad
    \begin{minipage}[b]{0.45\textwidth}
        \centering
        \begin{tikzpicture}[>=stealth, thick]

\def\A{1.6}   
\def\T{1.8}   

\draw[->, line width=1pt] (0.3,0) -- (7.1,0);
\draw[->, line width=1pt] (0.7,-2.4) -- (0.7,2.6);

\draw[very thick, draw=black, samples=400, smooth, domain=1.1:6.38]
  plot (\x, {1.6*sin(180*\x*\x/6 + 2*\x)});

\draw[dashed, black] (0.7,\A) -- (6.5,\A);
\draw[dashed, black] (0.7,-\A) -- (6.5,-\A);

\filldraw[black] (\T/4+1.24,\A) circle (3pt);
\node (x0) at (\T/4+1.24,2) {$x_1$};

\filldraw[black] (\T/4+3.38,\A) circle (3pt);
\node (x0) at (\T/4+3.38,2) {$x_3$};

\filldraw[black] (\T/4+4.7,\A) circle (3pt);
\node (x0) at (\T/4+4.7,2) {$x_5$};

\filldraw[black] (\T/4+5.76,\A) circle (3pt);
\node (x0) at (\T/4+5.76,2) {$x_7$};

\filldraw[black] (\T/4+2.53,-\A) circle (3pt);
\node (x0) at (\T/4+2.6,-2) {$x_2$};

\filldraw[black] (\T/4+4.11,-\A) circle (3pt);
\node (x0) at (\T/4+4.16,-2) {$x_4$};

\filldraw[black] (\T/4+5.26,-\A) circle (3pt);
\node (x0) at (\T/4+5.31,-2) {$x_6$};

\node (a1) at (0.4,1.6) {$1$};
\node (a2) at (0.25,-1.6) {$-1$};

\node (b) at (2*\T,3.3) {{\color{black}The plot of~$\displaystyle \frac{g(x)-f(x)}{E}$}};

\draw[black] (2*\T-2.3, 3.3-.6) rectangle (2*\T+2.3, 3.3+.6);

\end{tikzpicture}

    \end{minipage}
    }
    \caption{Illustration of the Chebyshev Equioscillation Theorem.}
    \label{fig:The Equioscillation Theorem}
\end{figure}
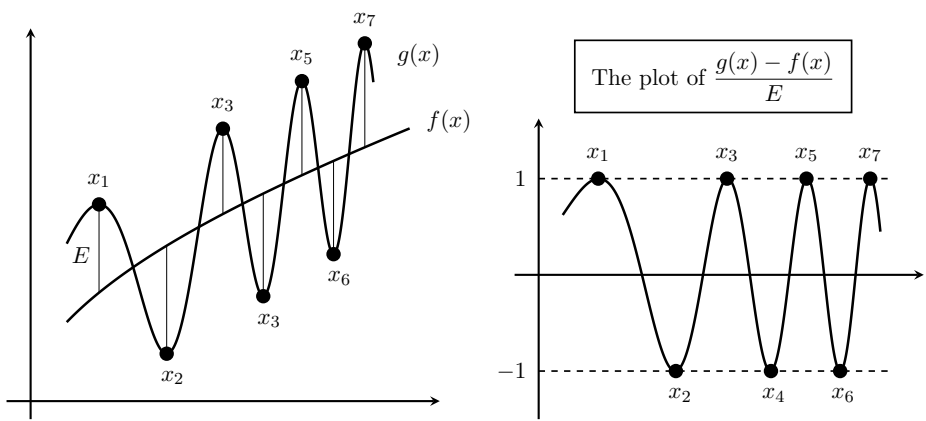

\begin{rmk}
Figure \ref{fig:The Equioscillation Theorem} is a simple illustration of the Chebyshev Equioscillation Theorem. The minimax polynomial $g$ is unique. The alternation property also determines $g$ and $E$ uniquely. But the points $x_1,x_2,\ldots,x_{d+2}$ need not be unique and depend on $f$ and $d$. For certain classical systems (e.g., Chebyshev polynomials), they are known explicitly; in general, they can be computed numerically (e.g., by the {\em Remez algorithm} \cite{fraser1965survey}, which costs at most $O(d^3)$).
\end{rmk}

Polynomial approximation provides a fundamental approach to representing functions. For a given function $f:[-1,1]\to\mathbb{R}$, one may study the best degree-$d$ polynomial approximation in the uniform norm, a classical problem in approximation theory. From the perspective of computational complexity, a more relevant quantity is the approximate degree: the minimal degree of a polynomial that approximates $f$ within error $\varepsilon$ in the uniform norm (see the definition below). This quantity captures the intrinsic complexity of approximating $f$ and plays a central role in various computational models. 

Computing the approximate degree is generally challenging, as it involves optimization over an integer-valued parameter rather than a continuous one. A powerful technique for analyzing approximate degree is the dual polynomial method~\cite{bun2022approximate}, which provides certificates for lower bounds and has been widely used in complexity theory.

Since quantum algorithms typically involve unitary or bounded operators whose spectra lie in $[-1,1]$, we restrict our attention to functions defined on subintervals of $[-1,1]$.

\begin{definition}[Approximate degree]
\label{defn:approximate degree}
Let $f(x): [a,b] \subseteq  [-1,1] \rightarrow [-1,1]$ be a function and let $\varepsilon\in(0,1]$. We define the approximate degree of $f$ as
\be \label{app deg}
\widetilde{\deg}_{\varepsilon}(f) = \min\{d \in \mathbb{N}: |f(x) -g(x)| \leq \varepsilon 
\text{ for all } x\in [a,b], \deg(g)=d \}.
\ee
\end{definition}

We establish below a connection between the dual polynomial method and the Chebyshev Equioscillation Theorem, which might be of independent interest.
For any fixed $d$, the dual polynomial method considers the following linear semi-infinite program (LSIP):
\bea
\label{original LP}
\min_{g,\delta} && \delta  \\
\text{s.t.} && |f(x)-g(x)| \leq \delta, \quad \forall x\in[a,b], \\
&& \deg(g) \leq d .
\eea
By known results from LSIP \cite{lai1992linear,shapiro2009semi}, there exist $\{x_1,\ldots,x_{d+2}\}$  such that (\ref{original LP}) is equivalent to the following linear program (here $d+2$ refers to the number of free parameters in the LSIP \eqref{original LP})
\bea
\label{original LP:discretisation}
\min_{g,\delta} && \delta  \\
\text{s.t.} && |f(x_i)-g(x_i)| \leq \delta, \quad \forall i=1,\ldots,{d+2}, 
\label{original LP:discretisatione-q1} \\
&& \deg(g) \leq d.
\eea

\begin{prop}
\label{prop for LSIP}
The points $\{x_1,\ldots,x_{d+2}\}$ are exactly the equioscillation points in the Chebyshev Equioscillation Theorem.
\end{prop}

\begin{proof}
See Appendix \ref{appA}.
\end{proof}

This result implies that in the Chebyshev Equioscillation Theorem, we can fix the degree as the approximate degree $d=\widetilde{\deg}_\varepsilon(f)$ for any desired accuracy $\varepsilon$, e.g., $\varepsilon=1/3$, we then will obtain a min-max polynomial approximation of $f(x)$ with error $\varepsilon$.

\subsection{Constructing periodic Jacobi matrices from prescribed spectra}

In this part, we establish a connection between periodic Jacobi matrices and polynomial approximation. The following is a known result about periodic Jacobi matrices.

\begin{prop}[Theorem 4.2 of \cite{ferguson1980construction}]
\label{prop:inverse probvlem of periodic Jacobi matrices}
Let $\Delta(x)$ be a polynomial of degree $m$. Then there is a periodic Jacobi matrix $A_0$ such that its discriminant is $\Delta(x)$ if and only if
\begin{itemize}
    \item the coefficient of $x^m$ in $\Delta(x)$ is positive, and
    \item $\Delta(x)$ has $m-1$ distinct local extrema $\nu_1>\cdots > \nu_{m-1}$ with $(-1)^j \Delta(\nu_j) \geq 1$ for all $j$.
\end{itemize}
Moreover, there is also an efficient algorithm to construct $A_0$ when given $\Delta(x)$.
\end{prop}

From \eqref{char poly of A}, suppose $\Delta(x)$ is the discriminant of $A_0$, then it is also the discriminant of $A_\theta$ for all $\theta$. Moreover, $1/e$ is the leading coefficient of $\Delta(x)$. It is easy to check that the second property in Proposition \ref{prop:inverse probvlem of periodic Jacobi matrices} is easily satisfied once $\Delta(x)$ has the oscillating property similar to the Chebyshev Equioscillation Theorem. This means that if $\Delta(x) \approx (f(x) - g(x))/E$ where $f(x), g(x), E$ are the same as those in the Chebyshev Equioscillation Theorem, then $\Delta(x)$ will be a discriminant of some periodic Jacobi matrix. In the Chebyshev Equioscillation Theorem, $g(x)$ is the polynomial with minimal degree $d$ that approximates $f(x)$ up to error $E$, so to ensure a nontrivial approximation, we have to choose $\Delta(x)$ such that $m=\deg(\Delta(x)) > \deg(g)$.

To continue, we recall the following concept from polynomial approximation theory.

\begin{definition}
Let $f(x): [a,b] \subseteq [-1,1]\rightarrow \mathbb{R}$ be a continuous function. For any integer $m$, define
\be \label{eq for Em}
E_m(f;[a,b]):=\min_{\deg(P)\leq m} \,\, \max_{x\in [a,b]} \,\, |f(x) - P(x)|
\ee
to be the best uniform (minimax) polynomial approximation error of degree at most $m$ over $[a,b]$. When there is no risk of confusion, we sometimes simply write it as $E_m$.
The polynomial of degree at most $m$ that attains this minimum is denoted as 
$P_m^*(x)$. 
\end{definition}

In practice, we start from $\varepsilon$, then compute the approximate degree $d=\widetilde{\deg}_{\varepsilon}(f)$, finally we have $E_{d-1}>\varepsilon\ge E_d$ and $d\le \widetilde{\deg}_{E_d}(f)$. Equality holds when $\varepsilon = E_d$. We indeed have $\widetilde{\deg}_\varepsilon(f)=\min\{d:E_d\leq \varepsilon\}$. Moreover, for any $\eta\in [E_d, E_{d-1})$, we always have $d= \widetilde{\deg}_\eta(f)$.

\begin{thm}
\label{thm:approxiamtion}
Let $f(x): [a,b] \subseteq [-1,1]\rightarrow \mathbb{R}$ be a continuous function and $\varepsilon\in(0,1)$ with $d=\widetilde{\deg}_{\varepsilon}(f)$.
For any $\eta\in [E_d, E_{d-1})$, there is a periodic Jacobi matrix $A_\theta$ of the form \eqref{eq:periodic Jacobi matrix} with $\|A_\theta\|\leq 1$, 
dimension $d$ and discriminant $\Delta(x)$, such that 
\be \label{error bound}
\left| \Delta(x) - \frac{f(x) - P_{d-1}^*(x)}{E_{d-1}-\eta} \right| \leq 
\frac{\eta}{E_{d-1}-\eta}.
\ee
\end{thm}

\begin{proof}
Let $g(x)$ denote the degree $d$ polynomial in computing $d=\widetilde{\deg}_\eta(f)$.
Without loss of generality, we assume that the leading coefficient of $g(x)$ is positive; otherwise, $\Delta(x)$ below should be multiplied by the sign of this leading coefficient.
Denote
\begin{align}
\label{eq:def-Delta}
    \Delta(x) = \frac{g(x)-P_{d-1}^*(x)}{E_{d-1}-\eta},
\end{align} 
then it satisfies 
\begin{align}
\label{eq:decomposition}
    f(x) = (E_{d-1}-\eta) \Delta(x)+ P_{d-1}^*(x)+ (f(x)-g(x)).
\end{align}
By Theorem~\ref{thm:Chebyshev Equioscillation Theorem}, there exist $d+1$ points $a\le x_1<\cdots <x_{d+1}\le b$ and $\sigma\in \{-1,1\}$ such that $f(x_i)-P_{d-1}^*(x_i) = \sigma(-1)^i E_{d-1}$. 
Without loss of generality, we also set $\sigma=1$.
Since $|g(x_i)-f(x_i)|\le \eta < E_{d-1}$, by \eqref{eq:decomposition} the sign of $\Delta(x_i)$ satisfies  
\be
    \sgn(\Delta(x_i)) = \sgn\Big( (-1)^i E_{d-1}+g(x_i)-f(x_i) \Big) =  (-1)^i,
\ee
and hence oscillates between $-1$ and 1. By the Intermediate Value Theorem, there exist $d$ points $\{\xi_{i}\}_{i=1}^{d}$ satisfying $\Delta(\xi_i) = 0$ and $x_1<\xi_1<x_2<\xi_2<\cdots <\xi_{d}<x_{d+1}$. By Rolle's theorem, there exist $d-1$ points $\{y_i\}_{i=1}^{d-1}$ satisfying $\Delta'(y_i)=0$ and $\xi_1<y_1<\xi_2<y_2 <\cdots < y_{d-1} < \xi_{d}$. 
Since $\Delta'(x)$ is a polynomial with degree $d-1$ and has at most $d-1$ zeros, there are no other zeros other than $\{y_i\}_{i=1}^{d-1}$. Therefore, $\Delta(x)$ has oscillating signs on a sequence of intervals $(-\infty, \xi_1), (\xi_1,\xi_2), \ldots, (\xi_{d}, \infty)$. 
On each interval $(\xi_{i}, \xi_{i+1})$, $\Delta(x)$ attains the maximal/minimal value at $y_i$, so we have 
\be
     |\Delta(y_i)|\ge |\Delta(x_{i+1})| \ge \frac{E_{d-1}-|g(x_i)-f(x_i)|}{E_{d-1}-\eta}\ge\frac{E_{d-1}-\eta}{E_{d-1}-\eta}=1
\ee
as $x_{i+1}\in (\xi_{i}, \xi_{i+1})$. In conclusion, $\Delta(x)$ has $d-1$ extrema points 
$\{y_i\}_{i=1}^{d-1}$ with oscillating signs and $|\Delta(y_i)|\ge 1$. 
Now order the critical points from right to left by setting $\nu_j = y_{d-j}$ for $j=1,2,\ldots,d-1$. Then $\nu_1>\cdots>\nu_{d-1}$. Since the leading coefficient of $\Delta(x)$ is positive,  $\Delta(x) > 0$ for all $x \in(\xi_{d}, \infty)$. So $ (-1)^j\Delta(\nu_j) \ge 1$ for all $j$.
By Proposition~\ref{prop:inverse probvlem of periodic Jacobi matrices}, there exists a periodic Jacobi matrix $A_0$ with discriminant $\Delta(x)$. 
The claim \eqref{error bound} follows directly from \eqref{eq:decomposition}.

Finally, it remains to show $\|A_\theta\|\leq 1$. For any $c\in[-1,1]$, since the values $\Delta(x_i)$ alternate in sign and satisfy 
$|\Delta(x_i)|\geq 1$, we have $(\Delta(x_i)-c)(\Delta(x_{i+1})-c)\le 0$ for all $i$. 
Hence, each closed interval $[x_i, x_{i+1}]$ contains a root of $\Delta(x)-c$. Because $\Delta(x)-c$ has degree $d$, these are all its roots, and all of them lie in $[x_1,x_{d+1}] \subseteq [a,b] \subseteq [-1,1]$. Taking $c = \cos\theta$, we conclude that every eigenvalue of $A_\theta$ lies in $[-1,1]$.
\end{proof}

\begin{rmk}
\label{remark:key}
The above result is a key ingredient in the proofs of our main theorems. We need to choose $\varepsilon$ and $\eta$ as small constants, so that $E_{d-1}$ is also a constant. However, the bound for $E_d$ is not clear. The ideal situation for us is $E_d \ll E_{d-1}$. In this case, we can always find a constant $\eta$ such that $\eta/(E_{d-1}-\eta)$ is a small constant. Even in the worst case, $E_d$ is still a small constant, then as long as $E_d / E_{d-1} < 1/2$, we then have $E_d/(E_{d-1}-E_d)<1$, which is a small constant.
\end{rmk}

Many functions in practice satisfy the condition $E_{d}/E_{d-1}<1/2$ when $d$ is reasonably large. 
Asymptotically, when $f(x)$ is analytic and can be analytically continued to the open Bernstein ellipse $\mathcal{E}_\rho^0:=\{z=\frac{1}{2}(w+w^{-1}): |w|<\rho\}$, then $E_d \approx \rho^{-d}$; see \cite[Theorem 8.2]{Trefethen}.  Hence, heuristically, $E_d/E_{d-1} \approx 1/\rho$.
So the condition $E_d / E_{d-1} < 1/2$ is expected to hold when $\rho > 2$.

To be more exact, we consider three typical functions below as examples.
Recall that for continuous functions, the Chebyshev expansion provides (up to logarithmic factors) near-best polynomial approximation; see \cite[Theorem 16.1]{Trefethen}. In particular,
\be
E_d \leq \|f(x)-f_d(x)\|_\infty \leq c_d E_d,
\quad c_d=4+\frac{4}{\pi^2}\log(d+1),
\ee
where $f_d(x)$ is the degree-$d$ truncated Chebyshev expansion.
Below, we use this to calculate $E_{d}/E_{d-1}$. We denote $\widetilde{E}_d:=\|f(x)-f_d(x)\|_\infty$, and so
$E_{d}/E_{d-1} \leq c_{d-1} \widetilde{E}_{d}/\widetilde{E}_{d-1}$. 
\begin{itemize}
    \item {\bf Exponential function:} The function $f(x)=e^{-\beta x}$ for $\beta>0$ 
    appears in partition function computations \cite{brandao2008entanglement}. For this function, from the Chebyshev expansion, we have 
    $\widetilde{E}_d =\Theta(I_{d+1}(\beta))$, where $I_{d+1}(\beta)$ is the modified Bessel function of the first kind. Using the standard asymptotic behavior of Bessel functions, asymptotically,  when $d\gg \beta$ we have
    $E_{d}/E_{d-1}\leq c_{d-1} \widetilde{E}_{d}/\widetilde{E}_{d-1} = \Theta(\beta(\log d)/d)$, which is smaller than 1/2 when $d>\widetilde{\Omega}(\beta)$.

    \item {\bf Trigonometric function:} The function $f(x)=\sin(tx)$ appears as the imaginary part of $e^{itx}$ in Hamiltonian simulation \cite{gilyen2019qsvt}, where $t>0$ and $x\in[-1,1]$. It is an entire function with Chebyshev expansion $\sin(tx)=2 \sum_{k=0}^\infty (-1)^k J_{2k+1}(t) T_{2k+1}(x)$, where $J_{2k+1}(t)$ are Bessel functions of the first kind. So $\tilde{E}_d \approx \sum_{\text{odd } k>d} |J_k(t)|\approx |J_{d}(t)|$ if $d$ is odd. So when $d\gg t$,  we asymptotically have  $E_{d}/E_{d-1} \leq \Theta(t (\log d)/d)$, which is smaller than 1/2 when $d>\widetilde{\Omega}(t)$. A similar result holds for $\cos(tx)$.

    \item {\bf Logarithmic function:} The function $f(x) = \log(1+\beta x)$ with $0<\beta <1$ in the interval $[-1,1]$, is commonly used in the calculation of log-determinant \cite{han2017approximating}. For this function, we have $\widetilde{E}_d = \Theta(\rho^{-d}/d)$, where $\rho = \frac{1+\sqrt{1-\beta^2}}{\beta}$. Thus $E_{d}/E_{d-1} \leq \Theta( (\log d)/ \rho ) < 1/2$ when $\beta<O(1/ \log(d) )$.

    \item {\bf Inverse-type function:} The function $f(x)=1/\kappa x$ over the interval $[-1,-1/\kappa] \cup [1/\kappa,1]$ is used in solving linear systems \cite{gilyen2019qsvt}, where $\kappa>1$. 
    Now $\rho=\kappa+\sqrt{\kappa^2-1}\approx 2\kappa$.
    In this case, we have $E_d = \Theta(1/(2\kappa)^{d})$, and so $E_{d}/E_{d-1}\leq \Theta({\log (d)}/{\kappa})<1/2$ when $\kappa>\Omega(\log d)$. 
\end{itemize}

\section{Proof of Theorem \ref{thm:intro-main}}
\label{Sec: Proof of main Theorem}

\begin{thm}[Restatement of Theorem \ref{thm:intro-main}]
\label{thm:Restatement of Theorem1}
Let $\varepsilon<1$ be a small constant. Assume that $f(x):[a,b] \subseteq [-1,1]\rightarrow [-1,1]$ is a continuous function with $d:=\widetilde{\deg}_\varepsilon (f)=\Omega(\poly(n))$ and $E_{d}/E_{d-1}<1/2$.
Then the ``normalized trace estimation problem" is DQC1-complete.
\end{thm}


\begin{proof}
By Proposition \ref{prop: in DQC1}, the problem is in DQC1. 
Next, we show that it is DQC1-hard. Let $U=U_m\cdots U_2 U_1$ be a unitary acting on $n$ qubits such that each $U_i$ acts on at most $O(\log(n))$ qubits, with $m=O(\poly(n))$.
In the proof below, we first assume $m=d$ exactly, and will discuss how to remove this assumption at the end of the proof.

Consider
\be
A=
\sum_{t=1}^m a_t \ket{t-1} \bra{t-1} \otimes I +
\sum_{t=1}^m b_t \Big(\ket{t} \bra{t-1} \otimes U_t + \ket{t-1} \bra{t} \otimes U_t^\dag \Big),
\ee
where we set $\ket{m}=\ket{0}$.
Let $\ket{\phi_\theta}$ be an eigenstate of $U$, i.e., $U\ket{\phi_\theta} = e^{i\theta} \ket{\phi_\theta}$, and let
\bea
\ket{\psi_0(\theta)} &=& \ket{0} \otimes \ket{\phi_\theta}, \\
\ket{\psi_t(\theta)} &=& \ket{t} \otimes U_t\cdots U_1 \ket{\phi_\theta}, \quad \quad 
t=1,\ldots,m-1.
\eea
Then it is easy to check that
\bea
A \ket{\psi_0(\theta)} &=& 
a_1 \ket{\psi_0(\theta)} +
b_1 \ket{\psi_1(\theta)} + b_m \ket{m-1} \otimes U_m^\dag \ket{\phi_\theta}  \\
&=& a_1 \ket{\psi_0(\theta)} + b_1 \ket{\psi_1(\theta)} + b_m e^{-i\theta} \ket{\psi_{m-1}(\theta)} , \\
A  \ket{\psi_t(\theta)} &=& 
a_t \ket{\psi_t(\theta)} +
b_{t+1} \ket{\psi_{t+1}(\theta)} + b_{t-1} \ket{\psi_{t-1}(\theta)}, \quad\quad t=1,2,\ldots,m-2, \\
A  \ket{\psi_{m-1}(\theta)} &=& 
a_m \ket{\psi_{m-1}(\theta)} +
b_{m-1} \ket{\psi_{m-2}(\theta)} + b_{m} \ket{m} \otimes U_m \cdots U_1 \ket{\phi_\theta} \\
&=& a_m \ket{\psi_{m-1}(\theta)} +
b_{m-1} \ket{\psi_{m-2}(\theta)} + b_{m} e^{i\theta} \ket{\psi_{0}(\theta)}  .
\eea
Thus in the basis 
$\{\ket{\psi_t(\theta)}:t=0,\ldots,m-1\}$, the operator $A$ is a periodic Jacobi operator of the following form
\be
A_\theta := \begin{bmatrix}
a_1   & b_1 &           &        & b_me^{i\theta} \\
b_1  & a_2  & b_2       &        \\
     & b_2  &  \ddots   & \ddots \\
     &      & \ddots    & \ddots & b_{m-1}\\
b_m e^{-i\theta}        &&       & b_{m-1} & a_m \\
\end{bmatrix}.
\ee
Due to the orthogonality of $\{\ket{\psi_t(\theta)}\}_{t,\theta}$, we have that
\be
\label{1027:eq}
\tr[f(A)] = \sum_{\substack{\theta: \, e^{i\theta} \text{ is an} \\ \text{ eigenvalue of } U}} \tr[f(A_\theta)].
\ee

By Theorem \ref{thm:approxiamtion}, we can construct a periodic Jacobi matrix $A_\theta$ such that the dimension is $m=d$ and discriminant $\Delta(x)$. So the characteristic polynomial is $\det(xI_m-A_\theta)= e (\Delta(x) - \cos(\theta))$, where $e=2\prod_j b_j$. 
By \eqref{error bound}, we choose $\eta\in[E_d, E_{d-1})$ such that $E_{d-1}-\eta$ is a constant and
\be 
\left| \Delta(x) - \frac{f(x) - P_{d-1}^*(x)}{E_{d-1}-\eta} \right| \leq 
\frac{\eta}{E_{d-1}-\eta} =: \eta',
\ee
By Remark \ref{remark:key}, without loss of generality, we assume that $\eta'$ is a small constant when $E_{d}/E_{d-1}<1/2$.
The above implies that
\be
\Big| 
\tr[f(A_\theta)] - \tr[P_{d-1}^*(A_\theta)] - 
(E_{d-1}-\eta) \tr[\Delta(A_\theta)]  \Big| \leq (E_{d-1}-\eta) \eta'  m.
\ee
As a discriminant polynomial, we have that $\Delta(A_\theta) = \cos(\theta) I_m$, i.e., $\tr[\Delta(A_\theta)] = m \cos\theta$. So we have
\be \label{error 1}
\Big| 
\tr[f(A_\theta)] - \tr[P_{d-1}^*(A_\theta)] - 
(E_{d-1}-\eta) m \cos\theta \Big| \leq (E_{d-1}-\eta) \eta' m.
\ee

Since $P_{d-1}^*(x)$ has degree $d-1$, which is smaller than the size $m=d$ of $A_\theta$, we have that $\tr[P_{d-1}^*(A_\theta)]$ is independent of $\theta$ and can be computed efficiently and exactly. Thus, we can simply denote it as $\tr[P_{d-1}^*(A_0)]$ by setting $\theta=0$. One simple way to see this is using random walks. Indeed, $A_\theta$ can be viewed as the adjacency matrix of a circle of $m$ nodes. So, when doing random walks of length less than $m$, summing all the weights to compute the trace, the overall weight for vertex $i$ is the sum of weights of all walks that pass through the edge between vertex $1$ and vertex $m$ for even times, back and forth.

As a result of \eqref{1027:eq} and \eqref{error 1}, we have
\bea
\tr[f(A)] &=& \sum_{\substack{\theta: \, e^{i\theta} \text{ is an} \\ \text{ eigenvalue of } U}}  
\Big (m (E_{d-1}-\eta) (\cos \theta \pm \eta') + \tr[P_{d-1}^*(A_0)] \Big) \\
 &=& m (E_{d-1}-\eta) (\Re \tr{U} \pm 2^n \eta') + 2^n \tr[P_{d-1}^*(A_0)].
\eea
Notice that the dimension of $A$ is $m2^n$, it follows that
\be \label{error2}
\frac{\tr[f(A)]}{m2^n}
= (E_{d-1}-\eta) \left( \frac{\Re\tr(U)}{2^n} \pm \eta' \right) + 
\frac{\tr[P_{d-1}^*(A_0)]}{m}.
\ee 
Thus, if we can estimate $\frac{\tr[f(A)]}{m2^n} \pm \xi$ for some small constant $\xi$, we then can estimate $\frac{\Re\tr(U)}{2^n} \pm \eta''$ with $\eta'' := \eta' + \frac{\xi}{E_{d-1} - \eta}$. Since $\eta', E_{d-1}-\eta$ are constants smaller than 1, we can choose $\xi$ small enough to ensure that $\eta''$ is also a constant smaller than 1. 
By Proposition \ref{prop: TrU is DQC1 complete}, estimating $ \frac{\Re\tr(U)}{2^n} \pm \eta''$ is DQC1-hard, so estimating $\frac{\tr[f(A)]}{m2^n}\pm \xi$ is also DQC1-hard.

In the above proof, we used the property that $m=d$. But $m$ is only determined by $U$, this condition can be easily satisfied if we have $\widetilde{\deg}_{\varepsilon}(f) >  m = \poly(n)$ because we can add identities into $U$. This is why we assumed that $\widetilde{\deg}_{\varepsilon}(f) = \Omega(\poly(n))$ in the theorem.
\end{proof}

\section{Lower bound on classical query complexity }

In this section, we apply the connection established above to prove an exponential lower bound on the classical query complexity for estimating the normalized trace of $f(A)$. This establishes an exponential quantum--classical separation.

Recall that in \cite[Theorem 1.1]{edenhofer2025dequantization}, given sparse access to an $s$-sparse Hermitian matrix $A$ with $\|A\| \le 1$, and a degree-$d$ polynomial $f(x)$ satisfying $|f(x)| \le 1$ for all $x \in [-1,1]$, there exists a classical algorithm with query complexity $O(s^d/\varepsilon^2)$ that outputs an $\varepsilon$-approximation of the normalized trace of $f(A)$. The oracles are standard: one oracle returns the nonzero entries of $A$, while the other returns the positions of the nonzero entries in each row. Basically, their algorithm works as follows: for each matrix monomial $A^k$, sparsity allows us to estimate one diagonal entry efficiently (for example, using the definition, which costs $O(s^k)$). Since the normalized trace is the expectation value of the diagonal entries, $O(1/\varepsilon^2)$ samples are enough to approximate it up to an additive error $\varepsilon$.

The result of \cite{edenhofer2025dequantization} also holds when $f(x)$ is continuous on $[-1,1]$. In this case, one may first approximate $f(x)$ by a degree-$d$ polynomial and then apply their result. Since the complexity depends on the degree, it is desirable that this approximation has low degree. The minimum degree required for such an approximation is known as the approximate degree $\widetilde{\deg}(f)$, defined in~\eqref{app deg}. In Theorem~\ref{thm:lower bound}, we prove a lower bound showing that the exponential dependence on the approximate degree is necessary.

To prove the lower bound, we first extend the $k$-Forrelation problem~\cite{aaronson2015forrelation} to the DQC1 query model and prove an exponential lower bound, which we then use to obtain the desired result.

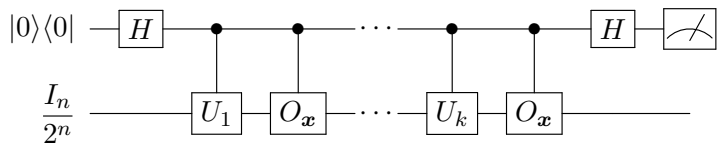
\begin{figure}[h]
\[
\Qcircuit @C=1em @R=1.5em {
\lstick{\ket{0}\bra{0}} &   \gate{H} & \ctrl{1}  & \ctrl{1}     & \qw&  \push{\kern-0.85em\mbox{$\cdots$}}  & \ctrl{1}  & \ctrl{1}      & \gate{H}  & \meter \\
\lstick{\displaystyle \frac{I_n}{2^n}}   &  \qw       & \gate{U_1} & \gate{O_{\x}} & \qw&  \push{\kern-0.85em\mbox{$\cdots$}} & \gate{U_k} & \gate{O_{\x}}  &  \qw  &  \qw  \\
}
\]
\caption{A $k$-query DQC1 algorithm.}
\label{circuit}
\end{figure}

The DQC1 query model has a structure similar to the BQP query model; see Figure~\ref{circuit}. Let $O_{\x}$ be an oracle defined by an input $\x=(x_1,\ldots,x_N) \in \{\pm 1\}^{N}$ via
\be
O_{\x}\ket{i} = x_i \ket{i},
\ee
where $N=2^n$. The goal is to compute a function of $\x$ using as few oracle queries as possible.
In a $k$-query DQC1 algorithm, we apply a unitary
\be
U = O_{\x} U_k \cdots O_{\x} U_2 O_{\x} U_1,
\ee
where $U_1,\ldots,U_k$ are independent of $O_{\x}$. A natural quantity that can be computed in this model is the normalized trace $\frac{1}{N}\tr[U]$.
We can also consider a more general setting in which different oracles $O_{\x_1}, O_{\x_2}, \ldots, O_{\x_k}$ are used, with inputs $\x_1, \ldots, \x_k \in \{\pm 1\}^N$.

We now consider an extension of the $k$-Forrelation problem in the DQC1 query model. Define the normalized trace function
\begin{equation}
{\tt Trace}_k(\x) :=
\frac{1}{N} \tr\!\left[
O_{\x_1} H^{\otimes n}
O_{\x_2} H^{\otimes n}
\cdots
H^{\otimes n}
O_{\x_k} H^{\otimes n}
\right],
\end{equation}
where $H$ is the $2\times 2$ Hadamard gate.
Recall that the function studied in the $k$-Forrelation problem corresponds to a single entry of the unitary
$H^{\otimes n} O_{\x_1} H^{\otimes n} \cdots H^{\otimes n} O_{\x_k} H^{\otimes n}.$
Thus, the above function can be viewed as a natural generalization of the $k$-Forrelation problem.

With Proposition \ref{prop: lower bound of trace 4}, we now prove the classical lower bound of normalized trace estimation.

\begin{thm}[Restatement of Theorem \ref{thm: intro lower bound}]
\label{thm:lower bound}
Let $s$ be an integer, and $\varepsilon<1$ be a small constant. Assume that $f(x):[a,b] \subseteq [-1,1]\rightarrow [-1,1]$ is a continuous function with $d:=\widetilde{\deg}_\varepsilon (f)=\Omega(\poly(n))$ and $E_{d}/E_{d-1}<1/2$.
Then there is an $s$-sparse, log-local Hamiltonian $A$ such that 
\be
\widetilde{\Omega} \left(
(s/2)^{(\widetilde{\deg}_\varepsilon(f)-4)/16}
\right)
\ee
classical queries are required to estimate the normalized trace of $f(A)$ up to a constant accuracy.
\end{thm}

\begin{proof}
We consider $k=4$ in ${\tt Trace}_k$. Now we have
\be
U = O_{\x_1} H^{\otimes n} O_{\x_2} H^{\otimes n} O_{\x_3} H^{\otimes n} O_{\x_4} H^{\otimes n} .
\ee
For convenience, we assume the sparsity $s=2^{r+1}$ and $n=r \ell$. The analysis below does not change too much if $r$ is not a factor of $n$.

We decompose $H^{\otimes n}$ as (this is the only place we used the result $n=r\ell$)
\be
H^{\otimes n} = (H^{\otimes r}\otimes I \otimes \cdots \otimes I) (I\otimes H^{\otimes r} \otimes \cdots \otimes I)  \cdots (I\otimes I \otimes \cdots \otimes H^{\otimes r}).  
\ee
There are $\ell$ factors in the above decomposition. Each factor is $2^r$ sparse. As a result, $U$ can be rewritten as a product of $4\ell$ $2^r$-sparse unitaries and $4$ 1-sparse unitaries. Namely, there are $m=4\ell+4$ unitaries in the decomposition of $U$.

Following the proof of Theorem \ref{thm:Restatement of Theorem1}, we can construct a Hamiltonian $A$, which is $2^{r+1}$-sparse and $O(r+\log n)$-local, such that (see \eqref{error2})
\be 
\frac{\tr[f(A)]}{m2^n}
= (E_{d-1}-\varepsilon) \left( \frac{\Re\tr(U)}{2^n} \pm \eta \right) + 
\frac{\tr[P_{d-1}^*(A_0)]}{m}.
\ee
Thus, if we can estimate $\frac{\tr[f(A)]}{m2^n} \pm \xi$ for some small constant $\xi$, we then can estimate $\frac{\Re\tr(U)}{2^n} \pm \eta'$ with $\eta' := \eta + \frac{\xi}{E_{d-1} - \varepsilon}$ a constant.
Moreover, $m=\widetilde{\deg}_\varepsilon(f)$, i.e., $\ell = (\widetilde{\deg}_\varepsilon(f)-4)/4$ and $n=r\ell = (\log (s/2))(\widetilde{\deg}_\varepsilon(f)-4)/4$. By Proposition \ref{prop: lower bound of trace 4}, we obtain a lower bound of estimating $\frac{1}{m2^n} \tr[f(A)]$, which is
\be
\widetilde{\Omega}\left( 2^{n/4} \right)
=
\widetilde{\Omega}\left( 2^{(\log (s/2))(\widetilde{\deg}(f)-4)/16}  \right)
=\widetilde{\Omega} \left(
(s/2)^{(\widetilde{\deg}(f)-4)/16}
\right).
\ee
This completes the proof.
\end{proof}

In comparison, the quantum algorithm only costs $\widetilde{O}(s\cdot \widetilde{\deg}_\varepsilon(f))$, see \cite{zhang2026low}. As a result, there is an exponential separation between quantum and classical algorithms for normalized trace estimation. 

Finally, inspired by the BQP-completeness of the forrelation problem in \cite{aaronson2015forrelation}, we show the following result, which might be of independent interest.

\begin{prop}[Restatement of Proposition \ref{prop:trace DQC1 complete-intro}]
Estimating ${\tt Trace}_{k}$ up to a constant accuracy is DQC1-complete. 
\end{prop}

\begin{proof}
From the proof of \cite[Theorem 5.2]{aaronson2015forrelation}, Aaronson and Ambainis indeed obtained the following strong result: Let $Q$ be a $n$-qubit quantum circuit over the universal gate set $\{H, {\rm CCSIGN}\}$, where ${\rm CCSIGN}$ is a three-qubit gate that maps $\ket{x,y,z}$ to $(-1)^{xyz} \ket{x,y,z}$. 
Assume there are $m$ Hadamard gates in $Q$.
Then there exist $f_1,\ldots,f_k:\{0,1\}^{n+1}\rightarrow \{1,-1\}$ with $k=O(m)$ such that $Q' = H^{\otimes (n+1)} O_{f_1} \cdots H^{\otimes (n+1)} O_{f_k} H^{\otimes (n+1)}$, where $Q'$ is either $(Q\otimes I)(I_n\otimes H)$ or $Q\otimes I$ depending on the number of appearance of the one dummy qubit in the construction. This implies that
\bea
\frac{1}{2^{n+1}} \tr[Q']= 
\frac{1}{2^{n}} \tr[Q] 
&=& \frac{1}{2^{n+1}} \tr[H^{\otimes (n+1)} O_{f_1} \cdots H^{\otimes (n+1)} O_{f_k} H^{\otimes (n+1)}] \\
&=& \frac{1}{2^{n+1}} \tr[O_{f_k f_1} H^{\otimes (n+1)} \cdots O_{f_{k-1}} H^{\otimes (n+1)} ] \\
&=& {\tt Trace}_{k-1}(f_kf_1,f_2,\ldots,f_{k-1}).
\eea
Since $Q$ comes from a universal gate set, estimating $\frac{1}{2^n} \tr[Q]$ is DQC1-complete. As a result, estimating ${\tt Trace}_{k}$ is also DQC1-complete. 
\end{proof}

\section{Conclusions}

In this work, we provide a general characterization of when estimating the normalized trace of functions of log-local Hamiltonians is DQC1-complete. We also prove an exponential classical lower bound for this problem. Together, these results provide a near-complete characterization of the computational hardness of normalized trace estimation in the DQC1 model. Moreover, both results identify the approximate degree as a key parameter governing the complexity of this problem.

We conclude with several interesting directions for future investigation:
\begin{enumerate}
    \item The DQC1 algorithm proposed in~\cite{cade} for estimating the normalized trace of $f(A)$ applies to any $K$-Lipschitz function with $K = \mathrm{poly}(n)$. This Lipschitz condition arises from the use of quantum phase estimation. Can the assumptions on the function in our main DQC1-hardness theorem \ref{thm:intro-main} be weakened? For example, can the result be extended to all continuous functions? We conjecture that this is the case. However, resolving this question would likely require new techniques to improve the approximation result in Theorem~\ref{thm:approxiamtion}, in particular to remove the condition $E_d/E_{d-1} < 1/2$. 

    \item It remains interesting to determine the optimal classical query complexity of $\mathrm{Trace}_k$. Proposition~1.5 gives an $\widetilde{\Omega}(N^{1/4})$ lower bound for $\mathrm{Trace}_4$, which is sufficient for establishing our exponential quantum--classical separation. However, this bound does not match the conjectured $\widetilde{\Omega}(N^{1/2})$ bound from Conjecture~1.4. More generally, resolving Conjecture~1.4 remains an interesting open problem.
\end{enumerate}

As studied in \cite{shao-gpt}, with the help of GPT 5.6 Sol, the condition $E_d/E_{d-1}<1/2$ can be relaxed to a Lipschitz condition. From our point of view, this means that the first question above is solved now. Regarding the second question, the lower bound in \cite{shepherd2010quantum} is not tight, but strong enough for our quantum-classical separation.
Resolving  Conjecture \ref{conj:trace} is still an interesting open problem for general $k$, even for $k=4$.

\section*{Acknowledgements}

ZJ is supported by the National Key Research Project of China (Grant No. 2023YFA1009403), the National Natural Science Foundation of China (Grant No. 12347104), and Beijing Science and Technology Planning Project (Grant No. Z25110100810000).
CS and YZ are supported by the National Key Research Project of China (Grant No. 2023YFA1009403 and 2025YFA1017200).
YZ is also supported by the European Research Council (ERC) via the Starting Grant QuanThermal (101222179).
TL and XZ are supported by the National Natural Science Foundation of China (Grant No. 62372006).
CS would like to thank Kewen Wu for helpful discussions on Conjecture \ref{conj:trace} and for suggesting the paper \cite{shepherd2010quantum}.

\appendix

\section{Proof of Proposition \ref{prop for LSIP}}
\label{appA}

This claim shows that in the Chebyshev Equioscillation Theorem, $d = \widetilde{\deg}_E(f)$.
We now prove this claim.
Note that the dual form of \eqref{original LP:discretisation} is 
\bea
\label{original LP:discretisation-dual}
\max_{h_i} && \sum_{i\in[{d+2}]} f(x_i) h_i \\
\text{s.t.} 
&& \sum_{i\in[{d+2}]} h_i x_i^{k} = 0, \quad \forall k\in\{0,1,\ldots,d\}, \label{dual-eq1} \\
&& \sum_{i\in[{d+2}]} |h_i| = 1. \label{dual-eq2}
\eea
The dual solution is 
\be
\label{optimal solution}
h_i = \frac{\alpha}{x_i\prod_{j\neq i} (x_i-x_j)}, \quad i\in[{d+2}],
\ee
where $\alpha \in \mathbb{R}$ is a free parameter, which is determined by the constraint (\ref{dual-eq2}).

To recover the optimal solution of the primal LP,  we assume that $g(x) = \sum_{k=0}^d g_k x^k$.  Let $g_k=g_k^+-g_k^-$ with $g_k^+, g_k^-\geq 0$. Then in matrix form, \eqref{original LP:discretisatione-q1} is
\be \label{LP-matrix from}
\begin{pmatrix}
1 & x_1^0 & \cdots & x_1^d & -x_1^0 & \cdots & -x_1^d \\
& & & \cdots\cdots\cdots \\
1 & x_{d+2}^0 & \cdots & x_{d+2}^d & -x_{d+2}^0 & \cdots & -x_{d+2}^d \\
1 & -x_1^0 & \cdots & -x_1^d & x_1^0 & \cdots & x_1^d \\
& & & \cdots\cdots\cdots \\
1 & -x_{d+2}^0 & \cdots & -x_{d+2}^d & x_{d+2}^0 & \cdots & x_{d+2}^d 
\end{pmatrix}
\begin{pmatrix}
\delta \\
g_0^+ \\
\cdots \\
g_{d}^+ \\
g_0^- \\
\cdots \\
g_{d}^- 
\end{pmatrix}
\geq 
\begin{pmatrix}
f(x_1) \\
\cdots \\
f(x_{d+2})  \\
-f(x_1) \\
\cdots \\
-f(x_{d+2})
\end{pmatrix}.
\ee
The objective function can be written as
\be
\delta = 
(1,0,\ldots,0,0,\ldots,0) \begin{pmatrix}
\delta \\
g_0^+ \\
\cdots \\
g_{d}^+ \\
g_0^- \\
\cdots \\
g_{d}^- 
\end{pmatrix}.
\ee
Therefore, the dual form has an objective function
\be
\sum_{i=1}^{d+2} f(x_i) (h_i^+ - h_i^-).
\ee
The constraints are described by the following 
\be
\begin{pmatrix}
1 & \cdots & 1 & 1 & \cdots & 1 \\
x_1^0 & \cdots & x_{\ell+2}^0 & -x_1^0 & \cdots & -x_{\ell+2}^0 \\
& &  \cdots\cdots\cdots \\
x_1^d & \cdots & x_{\ell+2}^d & -x_1^d & \cdots & -x_{\ell+2}^d \\
-x_1^0 & \cdots & -x_{\ell+2}^0 & x_1^0 & \cdots & x_{\ell+2}^0 \\
& &  \cdots\cdots\cdots \\
-x_1^d & \cdots & -x_{\ell+2}^d & x_1^d & \cdots & x_{\ell+2}^d
\end{pmatrix}
\begin{pmatrix}
h_1^+ \\
\cdots \\
h_{d+2}^+ \\
h_1^- \\
\cdots \\
h_{d+2}^- \\
\end{pmatrix} \leq 
\begin{pmatrix}
1 \\
0 \\
\cdots \\
0 \\
0 \\
\cdots \\
0 \\
\end{pmatrix}.
\ee
Namely,
\bea
&& \sum_{i=1}^{d+2} (h_i^+ + h_i^-) \leq 1 , \\
&& \sum_{i=1}^{d+2} (h_i^+ - h_i^-) x_i^k = 0.
\eea
Denote $h_i = h_i^+-h_i^-$, then we obtain the dual LP \eqref{original LP:discretisation-dual}.

WLOG, we assume that $x_1>\cdots>x_{d+2}>0$ and $\alpha>0$, then
\be
\begin{cases} \vspace{.2cm}
\displaystyle 
h_i^+ = \frac{\alpha}{x_i\prod_{j\neq i} (x_i^2-x_j^2)}, \quad h_i^- =0, & i \text{ is odd} \\
\displaystyle 
h_i^+ =0,  \quad h_i^- = -\frac{\alpha}{x_i\prod_{j\neq i} (x_i^2-x_j^2)}, & i \text{ is even} 
\end{cases}
\ee
From the complementary slackness between LP and dual LP, we know that if $h_i^+>0$ (or $h_i^->0$), then the corresponding constraint of the primal LP is tight, i.e., it is an equality in \eqref{LP-matrix from}. We now assume $d+2=2m$ is even for convenience. Then we obtain the following linear system:
\be 
\begin{pmatrix}
1 & x_1^0 & \cdots & x_1^d & -x_1^0 & \cdots & -x_1^d \\
& & & \cdots\cdots\cdots \\
1 & x_{2m-1}^0 & \cdots & x_{2m-1}^d & -x_{2m-1}^0 & \cdots & -x_{2m-1}^d\\
1 & -x_2^0 & \cdots & -x_2^d & x_2^0 & \cdots & x_2^d \\
& & & \cdots\cdots\cdots \\
1 & - x_{2m}^0 & \cdots & -x_{2m}^d & x_{2m}^0 & \cdots & x_{2m}^d
\end{pmatrix}
\begin{pmatrix}
\delta \\
g_0^+ \\
\cdots \\
g_{d}^+ \\
g_0^- \\
\cdots \\
g_{d}^- 
\end{pmatrix}
=
\begin{pmatrix}
f(x_1) \\
\cdots \\
f(x_{2m-1})  \\
-f(x_2) \\
\cdots \\
-f(x_{2m} )
\end{pmatrix}.
\ee
Using $g_j=g_j^+ - g_j^-$, we can simplify the above as follows:
\be  \label{final LP}
\begin{pmatrix}
1 & x_1^0 & x_1^2 & \cdots & x_1^d  \\
& & & \cdots\cdots\cdots \\
1 & x_{2m-1}^0 & x_{2m-1}^1 & \cdots & x_{2m-1}^d  \\
1 & -x_2^0 & -x_2^1 & \cdots & -x_2^d  \\
& & & \cdots\cdots\cdots \\
1 & -x_{2m}^0 & -x_{2m}^1 & \cdots & -x_{2m}^d
\end{pmatrix}
\begin{pmatrix}
\delta \\
g_0 \\
g_1 \\
\vdots \\
g_{d} \\
\end{pmatrix}
=
\begin{pmatrix}
f(x_1) \\
\cdots \\
f(x_{2m-1})  \\
-f(x_2) \\
\cdots \\
-f(x_{2m} )
\end{pmatrix}.
\ee
The above defines a linear system of $2m \times 2m$ with a unique solution. For $g(x)=g_0+g_1x+\cdots+g_{d}x^{d}$, we have $g(x_i) = f(x_i) + (-1)^i \delta$ from \eqref{final LP}. This coincides with the Chebyshev Equioscillation Theorem, and the Remez algorithm can find the points efficiently.

\begin{rmk}
In \cite[Theorem 1.10 of the arXiv version]{montanaro2024quantum}, it was proved that the ``entry estimation problem", i.e., estimating an entry of a function of a matrix, is BQP-complete. This result depends on the construction of a certain tridiagonal matrix from its spectra. The efficiency of the construction, which further depends on the finding of $x_1,\ldots,x_{d+2}$ described in the LP \eqref{original LP:discretisation}, is not discussed in that paper. From Proposition \ref{prop for LSIP}, we see that these points can be computed efficiently by the  Remez algorithm. This guarantees the efficiency of the construction of the tridiagonal matrix in the proof of the BQP-completeness.
\end{rmk}

\bibliographystyle{plain}
\bibliography{ref.bib}

\end{document}